\documentclass[12pt,a4paper,twoside]{article}
\usepackage{epsfig}
\usepackage{amsmath}
\usepackage{times}
\usepackage{helvet}
\usepackage{courier}
\usepackage{algorithm}
\usepackage{algorithmic}
\usepackage{amsthm}
\usepackage{amssymb}
\usepackage{url}
\usepackage{fullpage}
\usepackage{float}
\usepackage{verbatim}
\usepackage[font=scriptsize,labelfont=bf]{caption}
\usepackage{graphicx}
\usepackage{subfig}
\urldef{\mails}\path|{gutmant, noam}@cs.huji.ac.il|
\floatstyle{ruled}

\newtheorem{theorem}{Theorem}
\newtheorem{lemma}{Lemma}
\newtheorem{definition}{Definition}

\newtheorem{proposition}{Proposition}
\newtheorem{corollary}{Corollary}
\begin{document}
\newfloat{alg}{H}{lop}
\floatname{alg}{Algorithm}

\title{Fair Allocation Without Trade \\ Draft - Comment are welcome}

\author{Avital Gutman\thanks{Supported by a grant from the Israeli Science Foundation (ISF), and by the Google Inter-
university center for Electronic Markets and Auctions.} \and Noam Nisan$^*$ }
%\institute{The Hebrew University of Jerusalem, Israel \\
%\mails}

\date{}
%\author{Avital Gutman  and Noam Nisan}
%       \affaddr{The Hebrew University of Jerusalem}\\
%       \email{gutmant@cs.huji.ac.il}
%\and
%\author{Noam Nisan \\
%       \affaddr{The Hebrew University of Jerusalem}\\
%       \email{noam@cs.huji.ac.il}
%}

\maketitle

\begin{abstract}

We consider the age-old problem of allocating items among different agents in a way that is efficient and fair. Two papers, by Dolev et al. and Ghodsi et al., have recently studied this problem in the context of computer systems. Both papers had similar models for agent preferences, but advocated different notions of fairness.
We formalize both fairness notions in economic terms, extending them to apply to a larger family of utilities.
Noting that in settings with such utilities efficiency is easily achieved in multiple ways, we study notions of fairness as criteria for choosing between different efficient allocations.
Our technical results are algorithms for finding fair allocations corresponding to two fairness notions:
Regarding the notion suggested by Ghodsi et al., we present a polynomial-time algorithm that computes an allocation for a general class of fairness notions, in which their notion is included. For the other, suggested by Dolev et al., we show that a competitive market equilibrium achieves the desired notion of fairness, thereby obtaining a polynomial-time algorithm that computes such a fair allocation and solving the main open problem raised by Dolev et al.
\end{abstract}

\section{Introduction}
This paper deals with the classic question of allocating resources among different potential agents.
Specifically, we are interested in this question in the context of computer systems that need to
share their computational resources among different agents.
Resources in this context can be CPU time, main memory, disk space, communication links, etc.
The agents may be jobs, computers, or software agents representing them, and the allocation may be
implemented at the level of the network routers, the operating system, or by higher level software, whether centralized or distributed.

The departure point of this work is several recent attempts to look at the allocation problem in computer systems in abstract principled terms; by Dolev et al.\cite{NJC} and Ghodsi et al.\cite{DRF}.  In these papers, the basic model assumed that each agent desires a well-defined bundle of resources, and the allocation problem is to decide which fraction of his bundle each agent gets.
While the treatment is abstract, it is very clear that the motivation came directly from actual computer systems.
We wish to explicitly point out a key difference between the literature on allocating resources in computer systems
and the general economic literature on resource allocation, a difference we believe explains the near complete separation between the two:
The computing literature almost always assumes that each agent desires a well-defined bundle of resources, while the economic literature almost always considers the {\em trade off} that agents have between different resources.

The standard example of resource allocation by an operating system has each job requesting a well-specified set of resources
(e.g. 1000 CPUs with 1TB of main memory and 1GB/sec of communication bandwidth), and allocates among such requests.
We do not often see systems that can handle requests like ``either 1000 CPUs with 2TB main memory or 2000 CPUs with 1TB main memory''.
In fact, even when the underlying problem allows a trade-off between several possible bundles of resources, %in almost all levels of allocation in computerized systems
the allocation system usually first decides on a bundle for each agent and then attempts to allocate these chosen bundles to all agents.
An example is routing in a network, where the routing decision of choosing a path to the destination is in practice completely
decoupled from the bandwidth-allocation decision for the links on the chosen path.

On the other hand, the economics literature on resource allocation usually focuses on the trade-offs between different resources that are captured by
agent preferences. The fact that for some agent an apple may be a substitute to an orange results in a flexibility in preferences that allows sophisticated trade that can be beneficial to all parties (Figure \ref{fig:convind}).

The case where consumers' preferences do not allow any substitution between different goods is called the case of ``perfect complements'' (Figure \ref{fig:pcind}),
with the basic example being Leontief utilities: For any quantity vector $(x_1,...,x_m)$ of $m$ resources, $u$ is defined as $u(x_1...x_m) \! = \! min_j (x_j/r_j)$, where the $r_j$s are the relative proportions needed of the different goods \cite{AGT6}.
These are the utilities used (implicitly) by \cite{NJC,DRF} to capture the preferences in computer systems, and are the focus of attention of this paper.
However, the relations between resources are not necessary linear. It could be that some agents' demand of bandwidth would be in quadratic relation to number of CPUs, and that RAM would be proportional to $\log(DISK)$. The generalization to perfectly complementary functions is required to capture such relations.
\begin{figure} [h!]
\centering
 \subfloat[]{\label{fig:convind}\includegraphics[width = 0.4\textwidth]{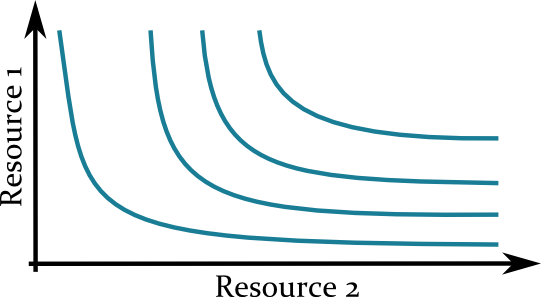}}
  \subfloat[]{\label{fig:pcind}\includegraphics[width = 0.4\textwidth]{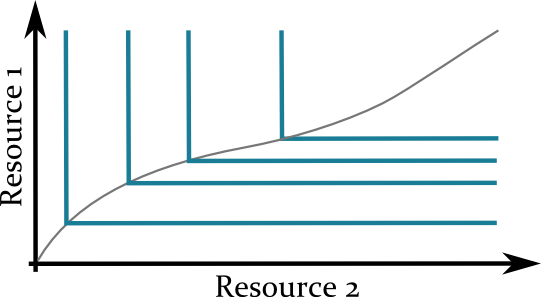}}

\caption{Figure \ref{fig:convind} shows an indifference map for a utility function with some degree of substitution between goods.
Figure \ref{fig:pcind} shows the indifference map of a perfectly complementary utility. Each indifference curve is the set of all allocations resulting in the same utility.}
\label{fig:indif}
\end{figure}

It is important to state that throughout this work, we do not assume or make any \emph{interpersonal comparisons of utility} \cite{interpersonal} - we use utility functions solely as a way to define agents' preferences over bundles.

Our first contribution in this paper is putting the question of allocation of computer resources studied in \cite{NJC,DRF} into an economic framework,
obtaining a general economic perspective that we believe is useful.
We observe that when agents have perfectly complementary preferences, the requirement of Pareto-efficiency turns out to be quite weak and simple:
it suffices that the allocation is (what we term) {\em non-wasteful}: no agent receives resources that he has no use for, and no useful resources are left on the table.
As a result, efficiency does not require any form of trade between agents, and there is no need for ``money'' as a mechanism of ensuring efficiency.
This underlying lack of trade may explain the
applicability to computer systems, which usually lack the infrastructure for enabling trade.
It may also suggest potential applicability in other scenarios where trade is impossible due to technical, administrative, legal, or ethical reasons.

However, this requirement alone leaves us with multiple possible allocations (as seen in Figure \ref{fig:edgpc} as opposed to Figure \ref{fig:edgconv}).
Therefore, the question of choosing between the efficient allocations becomes the central one, as indeed was done by \cite{NJC,DRF} in different ways, using different terminology and definitions of fairness.

\begin{figure} [h!]
\centering
\subfloat[With Substitutes]{\label{fig:edgconv}\includegraphics[width = 0.4\textwidth]{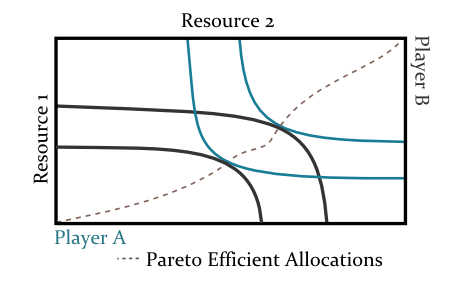}}
\subfloat[Perfect Complements]{\label{fig:edgpc}\includegraphics[width = 0.4\textwidth]{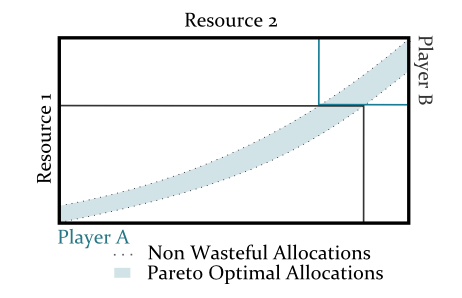}}
\caption{Figure \ref{fig:edgconv} is an Edgeworth box of continuous indifference curves with substitutes. In this case, if two curves are tangent to each other,
there is a single point of tangency. As a result, the Pareto set is single dimensional. \\
Figure \ref{fig:edgpc} is an Edgeworth box with curves corresponding to the perfectly complementary case.
Here, two curves may have many points of tangency, corresponding to allocations of the under-demanded goods.
As a result, the Pareto efficient allocations are in the whole shaded area.}

\end{figure}

After formalizing our notions, which generalize those in \cite{NJC,DRF} to the framework of perfectly complementary utilities,
we embark on the study of these fairness notions ``Bottleneck Based Fairness'' (BBF), advocated by \cite{NJC}, and ``Dominant Resource Fairness'' (DRF), used as the fairness property of the mechanism in \cite{DRF}.

Behind BBF there is the notion that when each agent is entitled to some percent of all resources, he has a ``justified complaint'' against an allocation if he gets less than that percentage on  all ``bottleneck'' resources.
An allocation is called BBF if no agent has a justified complaint against it.  Regarding DRF, we view the notion as composed of two
elements: The first element is a norm-like mechanism that quantifies and allows comparing bundles of resources allocated to a single
agent.
The second element dictates that fairness is defined according to
Rawlsian \cite{rawls} max-min fairness of these quantities.  Meaning, fairness should be decided by the agent that received the least according to the quantification of bundles.
DRF uses the $L_\infty$ norm (i.e. looking at the highest allocated share over all resources); we note any norm is equally viable. This general norm-like mechanism defines a class of fairness notions that we call ``Generalized Resource Fairness'' (GRF) (see examples in Figure \ref{fig:grf}). A somewhat related generalization of DRF can be found in \cite{egalitarian}.

At this point come our main technical results - both quite simple given the wider context which we have built.
We present two algorithms for finding fair allocations according to each of these fairness notions.
Our first algorithm finds an allocation satisfying any given fairness notion from the GRF class. It is similar, in concept, to the family of ``water-filling'' algorithms \cite{thomasCover,datanet}.
We extend the observation of \cite{DRF} that greedily allocating the resources to the
``poorest'' agent at each stage produces a fair allocation. In \cite{DRF} this was done in a pseudo-polynomial way
(allocating $\delta$-by-$\delta$ fractions of the goods\footnote{\cite{DRF} also suggest a polynomial-time algorithm which provides DRF for a discrete setup, a case which we do not study here.}) for DRF and Leontief utilities.
We give an algorithm which is polynomial for a wide subclass of perfectly complementary utilities and any GRF notion of fairness (given
minimal access to an oracle describing the utilities and desired fairness-norm), and becomes strongly polynomial for Leontief utilities.

Our second algorithmic result solves the main problem left open
by \cite{NJC}: obtaining a polynomial-time algorithm for computing a BBF allocation.
Our contribution here is a direct corollary of phrasing the problem in a market context, in which each agent enters the market with some quantity of each good (his endowment), or with a budget (some amount of money).
We recast the entitlements of \cite{NJC} in terms of budgets in a Fisher market\footnote{In terms of an Arrow-Debreu market \cite{ArrowDebreu},
this is the scenario in which agents have equal endowments of all goods.} (see e.g. \cite{AGT6}), and look at competitive market equilibria of this market.
Our main observation is then that any market equilibrium in this context will be BBF.  This both proves the existence of a BBF allocation for all perfectly complementary utilities, and solves the main problem of \cite{NJC} for Leontief utilities,
for which an equilibrium in a Fisher market can be found in polynomial-time using convex programming \cite{CV04}.
Note that Ghodsi et al. already compared their results with a market equilibrium.

The structure of the rest of the paper is as follows: Section 2 sets our notions, notations, and basic facts about efficiency with
perfectly complementary preferences.  In Section 3 we define and discuss various notions of fairness of allocation in this context,
in Section 4 we present the algorithm for GRF fairness and in Section 5 we develop the market context which implies the algorithm for BBF.

\section{The Model}

\subsection{Some Notation}

We use the following notation: For $\vec{x},\vec{z}\!\in\! \Re^m$ -
\begin{itemize}
\item We write $\vec{x} \leqslant \vec{z}$ if $\forall j$, $x_{j}\leqslant z_{j}$.
\item We write $\vec{x}<\vec{z}$ if $\vec{x} \leqslant \vec{z}$ and for some $j$ we have strict inequality: $x_{j} < z_{j}$.
\item We write $\vec{x}\ll \vec{z}$ if for all $j$, $x_j<z_j$.
\item Throughout the paper, $i\! \in \![n]$, $j\!\in\![m]$. $\vec{x}^i$ denotes the $i$th column vector of a matrix $X$.
\item Given a matrix $X\! = \!(x^i_j)\in \Re^{n \times m}$, $X \!+\! \vec{z}\! = \! (\vec{x}^1 \!+\! \vec{z},...,\vec{x}^n \!+\! \vec{z})$
\end{itemize}

\subsection{Basic Setup}

We study a setup with $m$ infinitely divisible resources, wanted by $n$ agents.
\begin{definition}
Given quantities for goods $\vec{q}\!\in\!\Re^m_+$, an \emph{allocation} is a matrix $X\!\in\!\Re^{n \times m}_+ $ with $\sum_i x_j^i\! \leqslant\! q_j$ for all $j$.
\end{definition}
Unless stated otherwise, we assume without loss of generality that $q_j\! = \!1$ for all $j$, in which case $x_j^i$ denotes the fraction of good $j$ that bidder $i$ gets.

Each bidder $i$ has a utility function $u_i(x^i_{1},...,x^i_{m})$ that denotes his utility from receiving the bundle composed of $x^i_{j}$ fraction of each resource $j$.   Our definitions and results are insensitive to the cardinal properties of $u_i$ and only depend on ordinal ones, so we could have equivalently modeled agent preferences using a preference relation $\prec_i$. Specifically, it means that we do not make any \emph{interpersonal comparisons of utility} \cite{interpersonal}. We make the standard assumption that utility functions are non-decreasing, that is - if $\vec{x}\!\ge\!\vec{y}$, then $u(\vec{x})\!\geqslant\! u(\vec{y})$. Additionally,  each agent has an entitlement $e_i$, with $\sum_i e_i \! = \!1$ , intuitively stating how much he brought into the system or ``deserves'' of the system. It is useful to think of an agent's entitlement as $1/n$, when $n$ is the number of agents, but, following \cite{NJC}, we also allow any pre-defined entitlements $e_i$ such that $\sum_i e_i \! = \!1$, where the general intention is that an agent with twice the entitlement of another one ``deserves'' twice the allocation.

\begin{definition}
$x$ is called \emph{maximal} for $u$ if for all $y$, $u(x)\geqslant u(y)$.
\end{definition}
\begin{definition}\label{def:lns}
A utility function $u$ is called \emph{strictly monotonic} if for every non-maximal $\vec{x}$, and all $\vec{y}>>\vec{x}$. We say it is \emph{non-satiable}\footnote{One can verify that strictly monotonic and non-satiable is equivalent to the standard economic term ``locally non-satiable for all $x$''.} if it is strictly monotonic and has no maximal $\vec{x}$.
\end{definition}

If $u$ is not \emph{non-satiable}, we say that it is \emph{satiable}, and $u$ is \emph{satiated at $\vec{x}$} if $u(\vec{x})\!\geqslant\!u(\vec{y})$ for all $\vec{y}\! \in\!\Re^m_+ $.

\subsection{Parsimonious Allocations}
For the utilities in question, it is often the case that an agent can relinquish some of his bundle without
reducing his utility. For clarity, and without losing generality, we will focus on allocations where this is not true.
We use $\vec{x} \Downarrow \vec{z}$ to denote a pointwise-minimum, that is $\vec{x} \Downarrow \vec{z} =  (min(x_1,z_1),...,min(x_m,z_m))$.
\begin{definition} A bundle $\vec{x}$ is a \emph{parsimonious bundle} for $u$ if for all $\vec{z}<\vec{x}$ we have that $u(\vec{z})<u(\vec{x})$. An allocation is called \emph{parsimonious} w.r.t utility functions $u_1,...,u_n$ if each $\vec{x}^i$ is a parsimonious bundle w.r.t to $u_i$.
\end{definition}
\begin{proposition}\label{prop:cont_pars}If $u$ is continuous, then for every bundle $\vec{x}$ there exists a parsimonious bundle $\vec{y}$  such that $\vec{y}\!\leqslant\! \vec{x}$ and $u(\vec{x})\! = \!u(\vec{y})$.
\end{proposition}

\begin{proof} Let $L\! = \!\left\{\vec{z}\in \Re^m_+| u(\vec{x})\! = \!u(\vec{z})\ and\ \vec{z}\leqslant \vec{x}\right\}$.
This set is compact (it is closed by the continuity of $u$ and bounded since $\vec{z}\leqslant \vec{x}$), so a $\vec{y} \in L$ that minimizes $\sum_j y_j$
over all $\vec{z} \in L$ exists.
Since $\vec{y}\in L$, we must have $\vec{y}\leqslant \vec{x}$ and $u(\vec{x})\! = \!u(\vec{y})$.
For any bundle such that $\vec{z}<\vec{y}$, it holds that $\vec{z} \leqslant \vec{x}$ and $\sum_j z_j<\sum_j y_j$. Hence, it must be the case that $u(\vec{z})<u(\vec{x})$,
otherwise $z \in L$ and $\vec{y}$ does not minimize $\sum_j z_j$ in $L$ . \end{proof}

\subsection{Perfect Complementarity}

\emph{Perfectly Complementary} utility functions could be defined in several ways.

\begin{definition}
A utility function $u$ is \emph{perfectly complementary} if for all $\vec{x},\vec{y}$ we have that $u(\vec{x} \Downarrow \vec{y}) \! = \! min\left(u(\vec{x}),u(\vec{y})\right)$.
\end{definition}

\begin{proposition} \label{prop:perf}The following are equivalent if $u$ is continuous:
\begin{enumerate}
\item $u$ is perfectly complementary.
\item For all parsimonious bundles $\vec{x},\vec{y}$, either $\vec{x}\leqslant \vec{y}$ or $\vec{y}\leqslant \vec{x}$.
\item There exists  a function $w:\Re_+  \rightarrow (\Re_+  \cup \{\infty\})^m$, called the {\em parsimonious bundle representation} of $u$,
such that for all $t \geqslant 0$ we have that
$u(\vec{x}) \geqslant t$ if and only if $\vec{x} \geqslant w(t)$.  That is, $u$ obtains utility level at least $t$ exactly when it gets at least $w_j(t)$
amount of every good $j$.
\end{enumerate}\end{proposition}
\begin{proof}
Assume $u$ is continuous.

{\bf $1 \Rightarrow 2$ ($\neg 2 \Rightarrow \neg 1$)}: Let $\vec{x},\vec{y}\in \Re^m_+ $ be two parsimonious bundles such that neither $\vec{x}\geqslant \vec{y}$ nor $\vec{y} \geqslant \vec{x}$. Let $\vec{z}\! = \!\vec{x}\Downarrow \vec{y}$. This means that there is some $j$ for which $x_j<y_j$ and some $l$ such that $x_l>y_l$. Clearly, $\vec{z} < \vec{x}$ and $\vec{z} < \vec{y}$. Since they are both parsimonious bundles, $u(\vec{z})<u(\vec{x})$ and $u(\vec{z})<u(\vec{y})$, and therefore $u(\vec{z})<min\left(u(\vec{x}),u(\vec{y})\right)$.

{\bf $2 \Rightarrow 3$}: Define $w(t)$ for all $t\geqslant 0$ to be a parsimonious bundle achieving utility level $t$ (and $\infty$ if such a bundle does not exist).
We need to show that $u(\vec{x})\geqslant t$ if and only if $\vec{x}\geqslant w(t)$.

($\Leftarrow$) Take a bundle $\vec{y}\geqslant w(t)$.  Then, from the monotonicity of $u$, $u(\vec{y}) \geqslant u(w(t)) \! = \! t$.

($\Rightarrow$) Let $\vec{y}$ be some bundle such that $u(\vec{y}) \geqslant t$ and let $\vec{z} \leqslant \vec{y}$  be a parsimonious bundle with $u(\vec{z})\! = \!u(\vec{y})$ (which exists by Proposition  \ref{prop:cont_pars}).
Since $\vec{z}$ and $w(t)$ are both parsimonious, by assumption, either $\vec{z} \geqslant w(t)$ or $\vec{z} \leqslant w(t)$. If it is the latter, $u(\vec{z})\leqslant u(w(t))\! = \!t$ which means $u(\vec{z})\! = \!t$ resulting in $\vec{z}\! = \!w(t)$ (otherwise, $w(t)$ is  not parsimonious).
Either way, $\vec{y} \geqslant \vec{z} \geqslant w(t)$, as required.

{\bf $3 \Rightarrow 1$}: We trivially have that $u(\vec{x} \Downarrow \vec{y}) \leqslant min\left(u(\vec{x}),u(\vec{y})\right)$ so we only need to show the opposite inequality.
Let $t\! = \!min(u(\vec{x}),u(\vec{y}))$. From (3) it follows that $\vec{x} \geqslant w(t)$ and $\vec{y} \geqslant w(t)$, and so $\vec{x} \Downarrow \vec{y} \geqslant w(t)$ so  $u(\vec{x} \Downarrow \vec{y}) \geqslant t$ as required.

This concludes the proof that the three definitions are equivalent for a continuous utility function $u$.  \end{proof}

\begin{lemma}\label{lemma:parsi} If $\vec{x}$ is a parsimonious bundle and $u$ is a perfectly complementary utility function, then $\vec{x}=w(u(\vec{x}))$\end{lemma}

\begin{proof}
Since $u$ is perfectly complementary, by Proposition \ref{prop:perf}, $\vec{x}\geqslant w(u(\vec{x}))$. Since $\vec{x}$ is a parsimonious bundle, for all $\vec{z}<\vec{x}$, $u(\vec{z})<u(\vec{x})$. Let $\vec{z}\! = \!w(u(\vec{x}))$. By definition, $u(\vec{z})\! = \!u(\vec{x})$.  However, if there exists some $j$ such that $x_j>z_j$, it contradicts $\vec{x}$ being parsimonious. \end{proof}

Note that Proposition \ref{prop:perf} implies that $w(t)$ must be non-decreasing.
However, the parsimonious bundle representation $w$ of a continuous perfectly complementary utility $u$ is not always continuous itself. It turns out that $u$ being strictly monotonic is equivalent to the continuity of $w$ (for all $t\in [0, max_{\vec{x}} u(\vec{x})]$).

\begin{lemma}\label{lemma:w_cont}
A continuous perfectly complementary utility $u$ is strictly monotonic if and only if its parsimonious bundle representation $w$ is continuous on the domain $\{t|0\leqslant t\leqslant max_{\vec{x}} u(\vec{x})\}$.\end{lemma}
\begin{proof}

(If):
Take a point $\vec{x}$ for which $u$ is not satiated, $u(\vec{x}) =t_x$, and some $\vec{y}$ such that $\vec{y} \gg \vec{x}$. Consider $w(t_x) \!+\! \varepsilon)$ for some $\varepsilon>0$, which is a parsimonious bundle with utility $t_x \!+\! \varepsilon > t_x$.
(Since $u$ is not satiated at $\vec{x}$, $\infty \gg w(t_x) \!+\! \varepsilon)$ for a small enough $\varepsilon$). Since $\vec{y}\gg \vec{x}=w(t_x)$ (by Lemma \ref{lemma:parsi} ), by continuity of $w$, for sufficiently small $\varepsilon$ we still have $\vec{y} \gg w(t_x) \!+\! \varepsilon)$, and so $u(\vec{y}) \geqslant t_x \!+\! \varepsilon > t_x$, as needed.

(Only if):
Let $(\alpha,z_{-j})$ be the vector obtained from replacing $\vec{z}$'s $j$th coordinate by $\alpha$. By definition, each coordinate of $w$ is non-decreasing, so to prove continuity it suffices to prove for every coordinate $j$ that there are no gaps in the values that $w_j(t)$ attains. Assume that a certain value $\beta \! = \! w_j(t)$ is achieved. We
need to show that for all $0 \leqslant \alpha < \beta$, there exists a parsimonious
bundle $w(t')$ such that $w_j(t')\! = \!\alpha$.

Take some bundle $\vec{z} \gg w(t)$. Let $t'\! = \!u(\alpha, z_{-j})$, and consider $w(t')$. From the monotonicity of $u$, and  Proposition \ref{prop:cont_pars}, it follows that $\vec{z} \gg w(t)\geqslant w(t')$. Since $w(t')$ is parsimonious, by definition, $w_j(t')\le\alpha$. It must be, therefore, that $w_j(t')\! = \!\alpha$, since otherwise, $w(t')\ll (\alpha,z_{-j})$, which contradicts $u$ being strictly monotonic. \end{proof}

In addition, perfectly complementary utilities are inherently quasi-concave.
\begin{definition}\label{def:sqc}$u$ is \emph{quasi-concave} if $u(\vec{x}) \geqslant u(\vec{y})$ implies that $u(\lambda \vec{x}  \!+\!  (1-\lambda) \vec{y})\geqslant u(\vec{y})$ for all $0<\lambda<1$.
It is \emph{strictly quasi-concave} if $u(\vec{x})>u(\vec{y})$ implies that $u(\lambda \vec{x}  \!+\!  (1-\lambda) \vec{y})>u(\vec{y})$  for all $0<\lambda<1$.
\end{definition}

\begin{lemma}\label{lemma:pcqc} Every perfectly complementary utility $u$ is  quasi-concave. \end{lemma}
\begin{proof}
Given a perfectly complementary utility function $u$, look at any pair $\vec{x},\vec{y}$ such that $u(\vec{x}) \geqslant u(\vec{y})$.
Let $u(x)\! = \!t_1$ and $u(y)\! = \!t_2$.
By Proposition 2, $\vec{x}\geqslant w(t_1)$ and $\vec{y}\geqslant w(t_2)$.
By Proposition \ref{prop:perf} and the monotonicity of $w$, $w(t_1)\geqslant w(t_2)$. Then, for all goods $j$ and any $\lambda\in (0,1)$ :

\begin{eqnarray*}
\lambda x_j  \!+\!  (1-\lambda) y_j &\ge& \lambda w_j(t_1)  \!+\!  (1-\lambda) w_j(t_2) \\
& \geqslant & \lambda w_j(t_2)  \!+\!  (1-\lambda) w_j(t_2) \\
& \! = \! & w_j(t_2)\end{eqnarray*}
Therefore, $u(\lambda \vec{x}  \!+\!  (1-\lambda) \vec{y})\geqslant u(w(t_2))\! = \! u(\vec{y})$. \end{proof}

If the utility function $u$ is also strictly monotonic, we can make the following stronger statement:

\begin{lemma}\label{lemma:pclns_sqc}
Every perfectly complementary and strictly monotonic utility function $u$ is strictly quasi-concave.
\end{lemma}

%\footnote{We will use strict quasi-concavity of functions in section 5.}

\begin{proof}
Let $\vec{x},\vec{y}$ be two bundles such that $u(\vec{x}) > u(\vec{y})$.
Let $u(\vec{x})\! = \!t_1$, $u(\vec{y})\! = \!t_2$, and $\vec{z}\! = \!\lambda \vec{x}  \!+\!  (1-\lambda)\vec{y}$.
Given $j$, if $ y_j\geqslant x_j$, then $y_j\geqslant z_j \geqslant x_j \geqslant\! w_j(t_1)$.
Otherwise, $x_j>z_j>y_j$, there is some $\hat{t}_j>t_2$ such that $w_j(\hat{t}_j)\! = \!z_j$ (Lemma \ref{lemma:w_cont}). For $\hat{t}\! = \!min(min_j(\hat{t}_j),t_1)$, $\hat{t}>t_2$ and therefore $\vec{z}\geqslant w(\hat{t})$. Hence, by definition $u(\vec{z})\geqslant \hat{t} > t_2 \! = \! u(\vec{y})$ as required. \end{proof}

\noindent {\bf Example: Leontief Functions}

Leontief utilities are one example of perfectly complementary utility functions.
These utility functions express an interest in a certain fixed {\em proportion of resources}.

$u$ is \emph{Leontief} if $u(\vec{x})\! = \!\min_{j\in S} (x_j/r_j)$ for some constants $\vec{r} \geqslant 0$, where $S \! = \! \{j | r_j > 0\}$.  This means that in order to obtain utility level $t\geqslant 0$, the agent needs at least $r_j t$ fraction of each good $j$. Thus, the parsimonious bundle for each utility level $t$ is $w(t)\! = \!(r_1 t,...,r_m t)$.

Leontief utilities are non-satiable.  The utilities considered in \cite{NJC} are satiable versions of Leontief utilities: $u(\vec{x})\! = \!\min(1, \min_{j\in S} (x_j/r_j))$, where $0 \! \leqslant \! r_j \! \leqslant \! 1$; i.e. the maximum utility of 1 is obtained at the minimal saturation bundle $(r_1 ... r_m)$.  We call these {\em satiable Leontief utilities}.

Note that a parsimonious bundle $\vec{x}$ for an agent with a Leontief utility needs to obtain the minimum on all coordinates, and therefore maintains $\vec{x} = \alpha \vec{r}$ for some constant $\alpha$.

\subsection{Efficiency}
The first requirement from an allocation is obviously to be efficient.
We start with a very weak notion of efficiency that intuitively does not assume the possibility of trade.
Though in general this notion of efficiency is significantly weaker than Pareto efficiency, we show that for perfectly complementary utilities it implies Pareto efficiency.

\begin{definition} An allocation is \emph{non-wasteful} if :
\begin{enumerate}
\item For all $i$, the bundle $\vec{x}^i$ is \emph{parsimonious} for $u_i$.
\item $\forall i\,\, u_{i}(\vec{x}^i \!+\! \vec{z})\! = \!u_{i}(\vec{x}^i)$
where $z_{j}\! = \!1-{\sum_{i}x_j^i}$.
\end{enumerate}
\end{definition}

Intuitively, these are allocations where no agent gets more unless it improves his utility, and goods are left on the table only if they cannot improve the agents' utilities.

Our interest is in economies where trade does not promote efficiency (the opposite of what is usual in economic theory).

Let $\vec{u}(X)\! = \!(u_1(\vec{x}^1),...,u_n(\vec{x}^n))$.
\begin{definition}
An allocation $X$ is \emph{Pareto efficient} if there is no other allocation $Z$  such that $\vec{u}(Z) > \vec{u}(X)$.
\end{definition}

Pareto-efficiency is stronger than being non-wasteful as it also requires that no bilateral (and multi-lateral) trade that is mutually beneficial is possible, while non-wasteful only requires goods not to be left on the table if they can be of any use to anyone. %, a condition that is in general much weaker.  %gutmant: This is redundant
However, for perfectly complementary utilities, non-wastefulness is a sufficient condition for Pareto-efficiency. To obtain equivalence we must add the requirement of parsimony to Pareto-efficiency. We can add it without loss of generality as we can replace any Pareto efficient allocation $X$ by a parsimonious allocation $Y \leqslant X$ that provides all agents with the same utilities.

\begin{definition}An economy $u_1,...,u_n$ is called a no-trade economy if every non-wasteful allocation is also Pareto efficient.\end{definition}

\begin{proposition} An economy composed of perfectly complementary utilities is a no-trade economy.\end{proposition}

\begin{proof}
Let $X$ be some non-wasteful allocation for $n$ perfectly complementary bidders.  Let $z_j\! = \!1-\sum_i x^i_j$. Since $X$ is non-wasteful, for all $i$, $u_i(\vec{x}^i)\! = \!u_i(\vec{x}^i \!+\! \vec{z})$. Assume, by way of contradiction, that there is an allocation $Y$ such that $\vec{u}(Y) > \vec{u}(X)\! = \!\vec{u}(X \!+\! \vec{z})$.
Therefore, there is some $i$ for which $u_i(\vec{y}^i)>u_i(\vec{x^i})$. We can assume, w.l.o.g, that $Y$ is a parsimonious allocation. By Proposition \ref{prop:perf}, $\vec{y}^i>\vec{x}^i$. Since $Y$ is an allocation, it must hold that $y^i-x^i\leqslant 1-\sum_k x^k$, but that contradicts $X$ being non-wasteful.
%Since the $u_i$s are perfectly complementary,  this is true if and only if $\vec{y}^i\geqslant \vec{x}^i \!+\! \vec{z}$ for all $i$, and for some $i$ the inequality is strict. Let $k$ be one such inequality, and let $j$ be a good for which $y^k_j>x^k_j \!+\! z_j$. We now notice that $1\! = \! z_j  \!+\!  \sum_i x^i_j < y^k_j  \!+\!  \sum_{i\ne k} x^i_j \leqslant \sum_i y^i_j$, which contradicts $Y$ being an allocation.
\end{proof}

\section{Fairness Notions}

Fairness of an allocation is an elusive concept. There are many opinions on what counts as fair, which sometimes vary depending on the application.
As a result, there are many fairness notions for allocations in various models.
Some recent papers dealing with the same motivation suggested different fairness notions.
One was introduced by Dolev et al. \cite{NJC} and the other by Ghodsi et al. \cite{DRF}.
In this section, we rephrase the definitions of \cite{NJC} and \cite{DRF},
generalize, and compare them.\footnote{In the case of Ghodsi et al.,
since the paper has a systems flavor, we use our understanding of their somewhat informal definitions.}

\subsection{Dominant / Generalized Resource Fairness}
The notion of ``Dominant Resource Fairness'' was advocated in \cite{DRF}.
Conceptually, we view this notion as combining two elements: the first is the
choice of defining fairness between the agents in the Rawlsian \cite{rawls} sense of maximizing the welfare of the poorest agent,
and the second  is the choice of quantifying the allocation of an agent according to the maximum share of any resource he got.
We generalize this notion of fairness by sticking to the first element of fairness and allowing a variety of choices for the second.

Let $|| \cdot ||$ be a norm on $\Re^m$. We think of $||\vec{x}^i||$ as our measure of ``how much'' $i$ obtains - a scalar that quantifies the bundle $\vec{x}^i$. In fact, we don't really need all the properties of a norm, just monotonicity and continuity.

\begin{definition} Given a norm $|| \cdot ||$ on $\Re^m$ and two parsimonious allocations $X,Y\in \Re^{n \times m}_+ $  (where $\vec{x}^i,\vec{y}^i$ are column vectors in $X,Y$ respectively).
We say that  $X$ is fairer than $Y$ (according to $|| \cdot ||$) if
\[\min_{i : ||\vec{x}^i|| \ne ||\vec{y}^i||} (||\vec{x}^i||) \geqslant \min_{i : ||\vec{x}^i|| \ne ||\vec{y}^i||} (||\vec{y}^i||)\]
That is, if we look at the agent with the minimal allocation that differs between $X$ and $Y$, then that agent gets (weakly) more in $X$.
Generalizing to the case that agents come with pre-defined entitlements, we also say that $X$ is fairer than $Y$ under
entitlements $e_1,...,e_n$ if
\[\min_{i : ||\vec{x}^i|| \ne ||\vec{y}^i||} (||\vec{x}^i/e_i||) \geqslant \min_{i : ||\vec{x}^i|| \ne ||\vec{y}^i||} (||\vec{y}^i/e_i||)\] We say that an allocation $X$ is \emph{$||\cdot||$-fair} if for every other allocation $Y$, we have that $X$ is fairer than $Y$ (And similarly for fairness under entitlements.).
\end{definition}

In \cite{DRF}, the $L_{\infty}$ norm was used, and the fairness notion was termed ``Dominant Resource Fairness'' (DRF).
Other natural choices would be the $L_1$ norm, referred to as ``Asset Fairness'' in \cite{DRF}, which counts total resource use,
and intermediate norms such as $L_2$.  The following example shows the different fair allocations for these three choices of norm:

\begin{figure} [h!]
\centering
 \subfloat[Request Matrix]{\label{fig:reqgrf}\includegraphics[width = 0.40\textwidth]{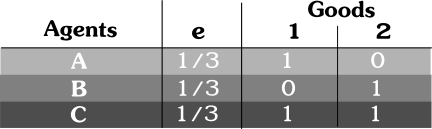}} \\

 \subfloat[$L_\infty$]{\label{fig:linfty}\includegraphics[width = 0.20\textwidth]{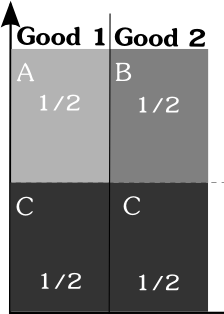}}
  \subfloat[$L_1$]{\label{fig:l1}\includegraphics[width = 0.20\textwidth]{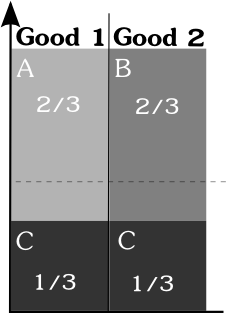}}
    \subfloat[$L_2$]{\label{fig:l2}\includegraphics[width = 0.20\textwidth]{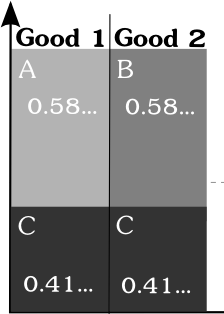}}
\caption{Figure \ref{fig:reqgrf} depicts the agents' interest in the resources and their respective entitlements. Figures (b)-(d) show the different allocations corresponding to three different norms. $0.41..$ is the solution to the equation $2x^2 = (1-x)^2$.}
\label{fig:grf}
\end{figure}

As seen in Figure \ref{fig:grf}, it is hard to tell which norm provides an intuitively fairer allocation, and each may be appropriate in some cases.
One may argue that $L_\infty$ is unfair since agent C receives too much in total resources, while $L_1$ is unfair as C receives too little of any specific resource.  $L_2$ gives an intermediate allocation, but then again, it may be argued that C still receives too much in total.

\subsection{No Justified Complaints}

This fairness notion was originally stated using Leontief utilities. We suggest a generalization of this fairness notion to perfectly complementary utilities:
\begin{definition} {\bf No Justified Complaints for PC utilities}
An allocation $X$ has the property of \emph{``No Justified Complaints''} (NJC) if for all agents $i$:
\begin{enumerate}
\item $\vec{x}^i$ is a parsimonious bundle.
\item Either there is some ``bottleneck'' good $j$ such that $\sum_i x^i_j\! = \!1$ with $x^i_j\geqslant e_i$, or
$u_i$ is satiated at $\vec{x}^i$: $u_i(\vec{x}^i) \geqslant u_i(\vec{z})$ for all $\vec{z}\in \Re^m$.
\end{enumerate}
\end{definition}
If an allocation has the NJC property, we will say it is \emph{``Bottleneck-Based Fair''} (BBF).

Restricting the above definition to satiable Leontief utilities reduces to the original definition in \cite{NJC}.

One can easily verify that the examples of $||\cdot||$-fair allocations above are all BBF. Of course, all criticism of these allocations applies to BBF as well. However, not all $||\cdot||$-fair allocations are BBF.
The examples in Figure \ref{fig:njc} show that BBF and norm-fairness are unrelated notions.

\begin{figure} [h!]
\centering
 \subfloat[Request Matrix]{\label{fig:reqnjc}\includegraphics[width = 0.4\textwidth]{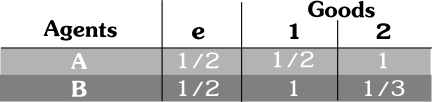}}\\
 \subfloat[$L_\infty$]{\label{fig:njclinfty}\includegraphics[width = 0.2\textwidth]{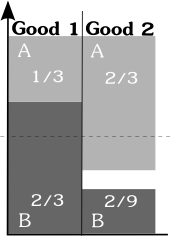}}
  \subfloat[$L_1$]{\label{fig:njcl1}\includegraphics[width = 0.2\textwidth]{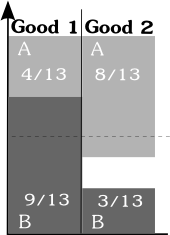}}
    \subfloat[$L_2$]{\label{fig:njcl2}\includegraphics[width = 0.2\textwidth]{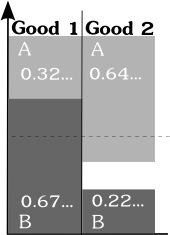}}
  \subfloat[BBF]{\label{fig:njcalloc}\includegraphics[width = 0.2\textwidth]{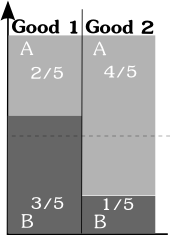}}
\caption{In the setup above, none of the $||\cdot||$-Fair allocations is BBF, since agent A does not have his entitlement of a bottleneck resource. (\ref{fig:njcalloc}) shows a BBF allocation of this setup, that does not satisfy any of the aforementioned norm-fairness notions.}
\label{fig:njc}
\end{figure}

\section{Computing a Norm-Fair Allocation}
%We provide a fully polynomial algorithm to find a $||\cdot||$-fair allocation in a setup with continuous, strictly monotonic and perfectly complementary utilities.

The basic idea of \cite{DRF} is that a greedy allocation rule approaches dominant resource fairness: at each stage the ``poorest'' agent gets another
$\varepsilon$ of his required bundle.  We observe that the same is true in general for all $||\cdot||$-fairness notions, and show how the pseudo-polynomial algorithm presented in \cite{DRF} may be converted into a fully polynomial-time allocation algorithm.
We will present the algorithm in general terms and assume the following minimal type of oracle access to the underlying utilities and norm.

Our algorithm applies to continuous, strictly monotonic, perfectly complementary utilities with the following additional property:

\begin{definition}\label{def:compatible}
Let $u$ be a perfectly complementary function, and $\vec{x},\vec{y}\in \Re^m_+ $ be two parsimonious bundles with respect to $u$.
We call $u$ \emph{compatible with a norm $||\cdot||$} if $u(\vec{x})>u(\vec{y}) $ implies $||\vec{x}||>||\vec{y}||$.\end{definition}

\begin{lemma}\label{lemma:comp}
Every perfectly complementary utility is compatible with $L_p$ for every $1 \leqslant p< \infty$.
Leontief  and satiable Leontief utilities are also compatible with $L_\infty$.\end{lemma}

\begin{proof}
Let $u$ be a perfectly complementary utility function, and let $\vec{x}, \vec{y}$ be two parsimonious bundles such that $u(\vec{x})>u(\vec{y})$. Let $u(\vec{x})\! = \!t$ and $u(\vec{y})\! = \!t'$.
Since $\vec{x},\vec{y}$ are parsimonious, and $u$ is monotonic, it follows that $\vec{x}\ge\vec{y}$ (Proposition \ref{prop:perf}). It cannot be the case that $\vec{x}\! = \!\vec{y}$ since $t>t'$, so $\vec{x}>\vec{y}$.

%Given $p$ such that $1\leqslant p < \infty$, we look at the $L_p$ norm of both bundles. %$ ||\vec{x}||_p\! = \!\sqrt[p] {\sum_j {x_j}^p} $ and $ ||\vec{y}||_p\! = \!\sqrt[p] {\sum_j {y_j}^p} $.

$\vec{x}>\vec{y}$, so $x_j\geqslant y_j$ for all $j$, with at least one inequality strict. Therefore, $\sum_j {x_j}^p > \sum_j {y_j}^p$
and from the monotonicity of $\sqrt[p]{}$, we obtain that $||\vec{x}||_p>||\vec{y}||_p$ as required.

For Leontief utilities, these parsimonious bundles are $t\cdot\vec{r}$ and $t'\cdot\vec{r}$, so in the case of $L_\infty$: $||t\cdot \vec{r}||_\infty \! = \! t \cdot max_j (\vec{r})>t'\cdot max_j (\vec{r}) \! = \!||t'\cdot \vec{r}||_\infty$,  as required.

The above applies trivially to satiable utilities as well, since for bundles at which $u$ is satiated, there will  not exist a bundle that gives a higher utility.  \end{proof}

Lemma \ref{lemma:comp} implies that if a perfectly complementary utility function is compatible with a norm, every $||\cdot||$-fair value corresponds to a single parsimonious allocation.

We need to describe how a utility function $u_i$ and a norm $|| \cdot ||$ are accessed. We assume that is provided by the following type of oracle query:
\begin{alg}Given a fairness level $h$, return the parsimonious bundle $\vec{x}^i$ for $i$ with $||\vec{x}^i||\! = \!h$.\end{alg}

Notice that this type of oracle access is no stronger than direct access to $u_i$ or its corresponding parsimonious representation $w^i$, as given access to $w^i(t)$ one can use binary search over $t$ to find a bundle with $||w^i(t)||\! = \!h$, and given only direct access to $u_i$ one can recover $w^i(t)$
by binary search on every index (i.e., $w^i_1(t)$ is the smallest $\alpha$ such that $u_i(\alpha, 1,...,1)=t$).  We prefer this formalization of access since it does not directly use the cardinal utilities but only their ordinal properties, and so will often be more natural.
For example, for Leontief utilities, where $u_i(\vec{x})\! = \! \min_j (x_j/r_j)$, the reply for query $h$ would be $\vec{x}\! = \!h\cdot \vec{r}/||\vec{r}||$.

We use this oracle in the following basic step of the algorithm:

\begin{alg}
{\bf Allocation Step(S,q)} (For a given $|| \cdot ||$) :

{\bf Input:}

A set S of agents and a vector of remaining quantities $q\! = \!(q_1 ... q_m)$.

{\bf Output:}
Returns the maximum $h$ such that each $i\in S$ can be given a parsimonious $\vec{x}^i$ with $||\vec{x}^i||\! = \!h\cdot e_i$ and
yet stay within the given quantities: for all $j$, $\sum_i x^i_j \leqslant q_j$.  It also returns the parsimonious allocation $\left(X\right)$
with $||\vec{x}^i||\! = \!h$.

{\bf Algorithm:} To find $h$, one may use binary search given the access to the utility functions specified above.

\end{alg}
Note that this algorithm allows us to get arbitrarily close to $h$, but cannot guarantee that $h$ is accurate. This step is therefore, in general, pseudo-polynomial.

For the special case of Leontief utilities where for each agent $i$, $u_i(\vec{x})\! = \!\min_j (x_j/r^i_j)$, the following is a direct solution: $h\! = \!\min_j (q_j \cdot ||\vec{r}^i||/\sum_i {e_i r^i_j})$.

There are two possible reasons that may limit $h$: either one of the agents gets satiated at the parsimonious bundle corresponding to $h$ or (at least) one of the items got exhausted.

\begin{lemma}
If all utilities are perfectly complementary, continuous, strictly monotonic and
compatible with $|| \cdot ||$ then AllocationStep returns an allocation where either for some $i\in S$, $u_i(\vec{x}^i) \geqslant u_i(\vec{y})$ for all $\vec{y}$ or for some $j$ for which $q_j>0$, $\sum_i x^i_j \! = \! q_j$.
\end{lemma}

\begin{proof}
Let $t_i \! = \! u_i(\vec{x}^i)$. Assume by way of contradiction that there is some $\varepsilon>0$ small enough such that for all $j$, $\sum_{i\in S} w^i_j(t_i \!+\! \varepsilon)\leqslant q_j$ (note that if $t_i$ is the maximal value of $u_i$, then $w_j^i(t_i \!+\! \varepsilon)\! = \!\infty$ for all $j$).
By Lemma \ref{lemma:w_cont}, it must be the case that $w^i(t_i \!+\! \varepsilon)>w^i(t_i)$. Since $u_i$ is compatible with $||\cdot||$ for all $i$, $||w^i(t_i \!+\! \varepsilon)||>||w^i(t_i)||\! = \!h$ for all $i$, which contradicts $h$'s maximality under the constraints.  \end{proof}
Now we have the following algorithm:

\begin{alg}
\begin{algorithmic}
\STATE Initialize $S$ to be all agents and $q$ to be initial  quantities.
\WHILE {$S\neq \emptyset$}
\STATE $h,\left(X\right) \! = \!$ AllocationStep$(S,q)$.
\STATE $G\leftarrow\{j | \sum_i x^i_j \! = \! q_j\}$ // Exhausted items
\STATE Let $F$ be the set of agents $i$ such that for all parsimonious bundles $\vec{y}$ with $||\vec{y}||>h$ we have $y_j > x^i_j$ for some $j \in G$.
\STATE $q \! = \! q - \sum_{i \in F} \vec{x}^i$
\STATE $S \! = \! S \setminus F$
\ENDWHILE
\end{algorithmic}
\end{alg}

\begin{theorem}This algorithm outputs the $||\cdot||$-fair allocation under entitlements $e_1,...,e_n$ for all continuous, strictly monotonic and perfectly complementary utility functions compatible with $||\cdot||$.\end{theorem}

\begin{proof}
First we notice that this algorithm terminates, since at every allocation step either some agent is satiated at the parsimonious allocation corresponding to $h$ (and then leaves $S$)
or some good is exhausted (and enters $G$), and since $S$ never increases and $G$ never decreases, we can have at most $n \!+\! m$ allocation steps.

We now prove correctness, by induction over the number of iterations of the main loop.
Let $Y$ be a parsimonious $||\cdot ||$-fair allocation, and $X$ be the allocation of the current step of the algorithm, and let $X'$ be the allocation of the previous step.
Similarly, let $S'$ be the set of non-exhausted agents in the end of the previous step, and $S$ be the set at the end of the current step.
Note that for all $i$, $\vec{x}^i\geqslant \vec{x}'^i$, since the algorithm never diminishes allocations.
We prove the following invariants by induction. The theorem follows directly.
\begin{enumerate}
\item If $i\in S$ then $||\vec{x}^i||\leqslant ||\vec{y}^i||$.
\item If $i \not\in S$ then $||\vec{x}^i||\! = \!||\vec{y}^i||$.
\end{enumerate}

{\bf Basis: } Right after the initialization phase, $\vec{x}^i\! = \!\vec{0}$ for all $i$ (which means $||\vec{x}^i||\! = \!0$) and $S\! = \![n]$)
so (1) holds trivially and (2) holds vacuously.

{\bf Step: } Assume that the above is true for $X'$. For all $i\not\in S'$ $||\vec{y}^i||\! = \!||\vec{x}'^i||\! = \!||\vec{x}^i||$, from the assumption and the fact that the algorithm
does not change allocations of agents not in $S'$.
Since we assumed $Y$ is $||\cdot||$-fair, for all $i\in S'$, $||\vec{y}^i||\ge||\vec{x}^i||\! = \!h$ since
otherwise $\min_{\left\{i : ||\vec{x}^i|| \ne ||\vec{y}^i||\right\}} ||\vec{x}^i||>\min_{\left\{i : ||\vec{x}^i|| \ne ||\vec{y}^i||\right\}} ||\vec{y}^i||$ and $Y$'s fairness is contradicted. Since $S \subsetneq S'$, (1) is proven.
Now, let's look at $i\in S'\setminus S$. It can't be that $||\vec{y}^i||>h$ since AllocationStep chooses the maximal $h$ for which the allocation is both
parsimonious and valid, so $||\vec{y}^i||\! = \!h$ for each $i\in S'\setminus S$.
We obtain therefore that for all $i\not\in S$, $||\vec{y}^i||\! = \!||\vec{x}^i||$ as required.

Observe that these two invariants imply that at the end of the algorithm, $X \! = \! Y$ for any $||\cdot||$-fair allocation, which not only proves the
correctness of the algorithm, but also the uniqueness of the fair allocation.  \end{proof}

\section{Competitive Equilibria and Bottleneck Based Fairness}
In this section we develop a close connection between the problem of finding a BBF allocation and the problem of finding an equilibrium in a Fisher Market,
a special case of an Arrow-Debreu market that is often more computationally tractable.  As corollaries of this connection and known results,
we get both a general existence result, generalizing that
of \cite{NJC}, and a polynomial-time algorithm for agents with Leontief utilities.

A Fisher Market is a market model where $n$ unrelated agents are interested in $m$ different, infinitely divisible goods,
each available in some quantity $q_j$. Each agent has a preference over bundles, given by a utility function $u_i$, and  a budget (``money'') that he
will use to trade with the other agents towards a better bundle.
The question of the existence and computation of equilibria in this setup has been
well studied in mathematical economics, and more recently in computational economics
(to state just a few references: \cite{ArrowDebreu,EISEN,DKV07}, see \cite{AGT6} for a recent survey).

\subsection{Fisher Market Equilibrium: Definition and Existence}
A market equilibrium is a state where all agents are no longer interested in trading with their peers.
In a Fisher Market, an equilibrium state is defined by the following two conditions (see \cite{AGT6}).
\begin{definition}
In a market setting with $n$ agents interested in $m$ infinitely divisible goods, where agents' preferences over bundles are given by a utility function $u_i$, and each agent has a budget $e_i>0$, an equilibrium is defined as a pair of price vector $\pi\in \Re^m_+ $ and an allocation $X\in \Re^{n\times m}_+ $ with the following two properties:
\begin{enumerate}
\item The vector $\vec{x}^i $ maximizes $u_i(\vec{x}^i)$ under the constraints $\sum_j \pi_j x^i_j \leqslant e_i$ and $\vec{x}^i\in \Re^m_+ $.
\item For each good $j\in [m]$, $\sum_i x^i_j\! = \!q_j$
\end{enumerate}
\end{definition}

It is imperative to understand that the two conditions are in some sense independent. When an agent calculates the bundle satisfying condition 1, he
{\bf does not} take into account the availability of goods in the market. If agents demand more than the available goods, the prices must go up - until an equilibrium is reached.

Fisher Market is a special case of the market studied by Arrow and Debreu, where all agents arrive to the market with the same proportions of all goods (\cite{AGT6}).
Thus, as a corollary of the Arrow and Debreu Theorem (\cite{ArrowDebreu}), given a Fisher market as described above, an equilibrium exists if the following conditions hold:
\begin{itemize}
\item $u_i$ is continuous
\item $u_i$ is strictly quasi-concave (Definition \ref{def:sqc}).
\item $u_i$ is non-satiable:\footnote{Note that if $u$ is strictly monotonic (Definition \ref{def:lns}) and there is no $\vec{x}$ in which $u$ is satiated, then it is non-satiable.}
For all $\vec{x}\in \Re^m_+ $, there is an $\vec{x}'\in \Re^m_+ $ such that $u_i(\vec{x}')>u_i(\vec{x})$
\end{itemize}

From Lemma \ref{lemma:pclns_sqc}, every perfectly complementary and strictly monotonic utility function is strictly quasi-concave, so we obtain the following corollary:
\begin{corollary} \label{exist_ce}
For a setup with continuous, strictly monotonic, non-satiable and perfectly complementary utility functions $u_i$, there exists a Fisher market equilibrium.
\end{corollary}

\subsection{The Non-Satiable Case}
The main point of this subsection is that market equilibrium for continuous perfectly complementary and {\em non-satiable} utilities is BBF.  The next subsection will observe that the non-satiability is not really required for strictly monotonic utilities.

 \begin{theorem} \label{thm:fishernjc}Let $u_1,...,u_n$ be perfectly complementary, continuous, strictly monotonic and non-satiable utility functions.
Let $e_1,...,e_n$ be the agents' budgets, let $(\pi,X)$ be a Fisher Market equilibrium in a market where $q_j=1$ for all $j$, and let $Y\leqslant X$ be the parsimonious allocation such that for all $i$,  $u_i(\vec{y}^i)\! = \!u_i(\vec{x}^i)$.  Then $Y$ is BBF.
\end{theorem}

\begin{proof}
Let $F\! = \!\{j|\pi_j\! = \!0\}$ be the ``free'' goods.  First notice that since the $u_i$s are strictly monotonic, continuous and non-satiable, we must have that all agents exhaust their budget, that is $\sum_j \pi_j x^i_j \! = \! e_i$ for all $i$, as otherwise the remaining money could then be used to buy a little bit extra of each good, thereby increasing $i$'s utility, contradicting the requirement that $\vec{x}^i$ maximizes $i$'s utility within budget.
As $\sum_i x^i_j \! = \! q_j\! = \!1$, it follows that $\sum_j \pi_j\! = \!\sum_j \sum_i x^i_j \pi_j\! = \!\sum_i \sum_j x^i_j \pi_j\! = \!\sum_i e_i\! = \!1$.
Now notice that for $j \not\in F$ we must
have $x^i_j\! = \!y^i_j$ for all $i$, since otherwise agent $i$ could demand only $y^i_j$ of
good $j$ saving $(x^i_j-y^i_j)\pi_j$ money, which could then be used to buy a little bit extra of each good.
Thus, $e_i \! = \! \sum_j \pi_j x^i_j \! = \! \sum_{j \not\in F} \pi_j x^i_j \! = \! \sum_{j \not\in F} \pi_j y^i_j$.  But since
$\sum_{j \not\in F} \pi_j \! = \! \sum_j \pi_j \! = \! 1$ we have a convex combination of $y^i_j$ over all $j \not\in F$ that sums to $e_i$, and it follows
that for some $j \not\in F$, $y^i_j \geqslant e_i$.  The proof is concluded by noting that all $j \not\in F$ are bottleneck resources in the allocation $Y$, since for all $j\notin F$, $1\! = \!\sum_i x^i_j \! = \! \sum_i y^i_j$. \end{proof}

\subsection{The Satiable Case}
Theorem \ref{thm:fishernjc} applies to almost all continuous, strictly monotonic and perfectly complementary utilities. The crux lies in the non-satiability assumption.
In this subsection, we show that this assumption is, in fact, without loss of generality.
Starting with a satiable, strictly monotonic, continuous and perfectly complementary function, we can convert it to a similarly characterized non-satiable function by adding a resource that only the agent whose utility is satiable is interested in.

Given a parsimonious bundle representation $w^i$ of a strictly monotonic, perfectly complementary and satiable utility function $u_i$,
and assuming, w.l.o.g, that the maximal value of $u_i$ is $1$ (otherwise, we normalize the function, without affecting any of its other properties), we define the parsimonious representation $w'^i$ of a continuous, perfectly complementary, strictly monotonic and non-satiable utility as follows:

\[
w'^i(t) \! = \! \left\{ \begin{array} {c c c}
\left(w^i_1(t),...,w^i_m(t),t\right)& & t \leqslant 1\\
t\cdot\left(w^i_1(1),...,w^i_m(1),1\right)& & t > 1\\
\end{array}
\right.
\]
Let $u'_i$ be the perfectly complementary utility function defined by $w'^i$.

\begin{proposition} If $u_i$ is continuous, strictly monotonic, and perfectly complementary, then $u'_i$ is continuous, strictly monotonic, perfectly complementary and non-satiable.
Moreover, if $u_i$ is satiable Leontief, then $u'_i$ is Leontief.
\end{proposition}

\begin{proof}
Since we defined $u'_i$ using the parsimonious bundle representation $w^i$, it is perfectly complementary by definition. It's also trivial to see that $u$ is non-satiable.
It remains to show it is continuous and strictly monotonic. Both are true for $t < 1$ as  wit$u'_i$ coincides with $u_i$ there;
both are true for $t>1$ as $u'_i$ is just a Leontief utility in that range; and at $t\! = \!1$ we get the same value from the two parts, thereby maintaining
both continuity and consistency.

If $u_i$ is satiable Leontief, then $u'_i$ could be otherwise written as:
$ u''_i \! = \! \min_{j\in [m \!+\! 1]}(x^i_j/r^i_j)$, where $r^i_{m \!+\! 1}\! = \!1$, and the remaining constants are unchanged.
This is clearly a standard Leontief function, and it is trivially equivalent to the function constructed above.  \end{proof}

\begin{lemma}\label{lem:satiable_njc}
 Let $u'_i$ be the extension of the satiable perfectly complementary utility $u_i$ for all $i$.  Then if an allocation is BBF with respect to the $u'_i$s it is also BBF with respect to the $u_i$s.
\end{lemma}

\begin{proof}
Given an allocation $X=(\vec{x}^1,...,\vec{x}^n)$ that is an equilibrium with respect to $(u'_1,...,u'_n)$, we need to show that the conditions of BBF hold with respect to $(u_1,...,u_n)$.
We know that for each agent $i$, there exists some $j\in J$ such that $x^i_j\geqslant e_i$, where $J\! = \!\left\{j | \sum_i x^i_j\! = \!1\right\}$ (since $u'_i$ has no maximum, the other option is irrelevant).

To show that with respect to the original utilities and setup the very same allocation is BBF, we have to discard all ``virtual resources'' added when extending the utilities.

Let $J\! = \!\{j | \sum_i x^i_j \! = \! 1\}$ be the set of bottleneck resources in the given allocation. Let $J' \subseteq J$ be the set of bottleneck resources after the virtual ones have been discarded.
Given an agent $i$, if $J$ does not include the ``virtual resource?€™?€™ added for $i$, then after the discard there still  exists $j\in J'$ such that $x^i_j\geqslant e_i$ as required.
If $J$ does include $i$'s ``virtual resource?€™?€™, then it is a bottleneck. This means that agent $i$ got everything available of that resource and therefore, by construction of the utility function $u'_i$, $x^i_j \! = \! w^i_j (\max (u_i(\vec{x})))$ for all $j\in J$ (meaning, agent $i$ got all he asked for).  \end{proof}

\subsection{The Bottom Line}
Looking at what we have constructed so far, we now have the tools to generalize \cite{NJC}.

\begin{theorem}
In every setup with perfectly complementary, continuous, and strictly monotonic utility functions $u_1,...,u_n$ there exists an allocation that satisfies ``No Justified Complaints''.
\end{theorem}

\begin{proof}
If needed, we extend the $u_i$s to be non-satiable as shown in the previous section.  Then take a competitive equilibrium of the $u'_i$s, guaranteed to exist by Corollary \ref{exist_ce}, and then by Theorem \ref{thm:fishernjc} the parsimonious bundles derived from that allocation are BBF for the $u'_i$s, and thus by Lemma \ref{lem:satiable_njc} also for the $u_i$s. \end{proof}

\begin{theorem}
There exists a polynomial-time algorithm that computes a \emph{``Bottleneck-Based Fair} allocation for every setup with Leontief or satiable Leontief utilities.
\end{theorem}
\begin{proof}

Codenotti and Varadarajan \cite{CV04} introduces an algorithm based on convex programming that finds a competitive equilibrium in a
Fisher Market for Leontief utilities.  In fact, the solution to their convex program provides the relevant parsimonious bundles. For the case of satiable Leontief utilities, the process of converting them to Leontief utilities described in the previous section is algorithmically trivial. Thus, we obtain a polynomial-time algorithm that finds a BBF allocation in every setup with Leontief or satiable Leontief utilities.  \end{proof}

\section{Discussion}
As computer systems become larger and more pervasive, servicing more and more agents, the question of fair resource allocation becomes more relevant. The classic requirements of an operating system's scheduler will no longer do.

We have studied two notions of fairness: DRF \cite{DRF} and BBF \cite{NJC}. The fact that both used the same family of utility function seems to indicate something about the nature of utility functions in such systems. Therefore, we have extended it to a larger family of utility functions - which may apply to other scenarios as well. The generalization of DRF to GRF stresses how vague the concept of fairness is, as it requires a few independent decisions, none of which comes with a strict sense of what is right and what is wrong.
Establishing the connection between the question of BBF and Fisher Market equilibrium shows that answers can sometimes be found in unexpected places. It also implies that any solution to the market equilibrium problem for continuous, strictly monotonic and perfectly complementary utilities automatically produces a BBF allocation (for example, \cite{DKV07}).
An important issue for further research is the question of incentive. For example, \cite{DRF} have shown that their algorithm for DRF is strategy proof and envy-free. In addition, the polynomial-time algorithm suggested here for BBF may allow us to prove or disprove the existence of similar properties for this fairness notion.

\bibliographystyle{plain}%2
\bibliography{fishernjc}
\end{document}